%% file: root.tex
\documentclass[11pt]{article}

\usepackage[cp1251]{inputenc} 

\usepackage{amsmath}
\usepackage{amssymb}
\usepackage{amsthm}
\usepackage{latexsym}
\usepackage[dvips]{graphicx,color}
\usepackage{float}
\usepackage{array}
\usepackage{xcolor}
\usepackage{enumerate} 
\usepackage[shortlabels]{enumitem}
\usepackage{titlesec}
\usepackage{multirow}
\usepackage[hyphens]{url}
\usepackage{hyperref}
\usepackage{etoolbox}

\titleformat{\section}[block]{\Large\sc\filcenter}{\thesection.}{5pt}{}
\titleformat{\subsection}[block]{\sc\filcenter}{\thesubsection.}{5pt}{}

\usepackage[square,sort,comma,numbers]{natbib}
\theoremstyle{plain}
\newtheorem{theorem}{Theorem}[section]

\newtheorem{lemma}[theorem]{Lemma}

\newtheorem{problem}[theorem]{Problem}

\theoremstyle{definition}
\newtheorem{definition}[theorem]{Definition}

\theoremstyle{remark}

\input{2018-macros}

\begin{document}

\title{Revisiting the generalized {\lt}  theorem}

\author{Abhisekh Sankaran}
\date{Department of Computer Science and Technology,\\
University of Cambridge}

\maketitle

\begin{abstract}
  We present a new proof of the generalized {\lt} theorem ($\glt{k}$)
  introduced in~\cite{abhisekh-apal}, over arbitrary
  structures. Instead of using $\lambda$-saturation as
  in~\cite{abhisekh-apal}, we construct just the ``required
  saturation'' directly using ascending chains of structures. We also
  strengthen the failure of $\glt{k}$ in the finite shown
  in~\cite{wollic12-paper}, by strengthening the failure of the {\lt}
  theorem in this context. In particular, we prove that not just
  universal sentences, but for each fixed $k$, even $\Sigma^0_2$
  sentences containing $k$ existential quantifiers fail to capture
  hereditariness in the finite. We conclude with two problems as
  future directions, concerning the {\lt} theorem and $\glt{k}$, both
  in the context of all finite structures.
\end{abstract}

\input{intro-5.0}
\input{background}

\input{proof-of-GLT-using-ascending-chains-2.0}

\input{strong-failure-of-LT-in-finite}
\input{conclusion-2.0}

%%%%%%%%%%%%%%%%%%%%%%%% Bibliography %%%%%%%%%%%%%%%%%%%%%%%%%%%%%%

\bibliographystyle{plain}

\bibliography{Bibfiles/model-theory,Bibfiles/finite-model-theory,Bibfiles/databases-and-verification,Bibfiles/automata,Bibfiles/logic-over-special-graph-classes,Bibfiles/structural-graph-theory,Bibfiles/algorithms,Bibfiles/misc,Bibfiles/self}

%%%%%%%%%%%%%%%%%%%%% ---------------- %%%%%%%%%%%%%%%%%%%%%%%%%%%%%

%%%%%%%%%%%%%%%%%%%%%%%%%% Appendices %%%%%%%%%%%%%%%%%%%%%%%%%%%%%%

%%%%%%%%%%%%%%%%%%%%%%% ---------------- %%%%%%%%%%%%%%%%%%%%%%%%%%%

\end{document}

%% file: 2018-macros.tex
\newcommand{\ti}[1]{\textit{#1}}
\newcommand{\tbf}[1]{\textbf{#1}}
\newcommand{\mc}[1]{\mathcal{#1}}
\newcommand{\und}[1]{\underline{#1}}
\newcommand{\mf}[1]{\mathfrak{#1}}

\newcommand{\fo}{\ensuremath{\text{FO}}}

\newcommand{\univ}[1]{\ensuremath{\mathsf{U}_{#1}}}

\newcommand{\lt}{{{\L}o{\'s}-Tarski}}

\newcommand{\glt}[1]{\ensuremath{\mathsf{GLT}({#1})}}

\newcommand{\eat}[1]{}

\newcommand{\cl}[1]{\ensuremath{\mathcal{#1}}}

%% file: intro-5.0.tex
\section{Introduction}\label{section:intro}

Preservation theorems are a class of results from classical model
theory that provide syntactic characterizations of first order (FO)
definable classes of arbitrary structures (structures that could be
finite or infinite), that are closed under given model-theoretic
operations. One of the earliest such results is the {\lt} theorem that
states that a class of arbitrary structures defined by an FO sentence
is hereditary (closed under substructures) if, and only if, it is
definable by a universal sentence (an FO sentence that contains only
universal quantifiers)~\cite{chang-keisler-short}. The theorem in
``dual'' form characterizes extension closed FO definable classes of
arbitrary structures in terms of existential sentences. The theorem
extends to theories (sets of sentences) as well. The {\lt} theorem is
historically important for classical model theory since its proof
constituted the earliest applications of the FO Compactness theorem (a
central result of model theory), and since it triggered off an
extensive study of preservation theorems for various other
model-theoretic operations (homomorphisms, unions of chains, direct
products, etc.), also for logics beyond FO (such as infinitary
logics)~\cite{hodges-history}.

Recently~\cite{abhisekh-apal}, a generalization of the {\lt} theorem
was proven by introducing and characterizing a new semantic property
that generalizes hereditariness in a parameterized manner.  We refer
to this property, called \emph{preservation under substructures modulo
  $k$-cruxes} in~\cite{abhisekh-apal}, as \emph{$k$-hereditariness} in
this paper. A class of structures is said to be $k$-hereditary if
every structure in the class contains a set of at most $k$ elements,
called a \emph{$k$-crux} of the structure, such that all substructures
(of the mentioned structure) \emph{containing the $k$-crux} are also
in the class.  For instance, consider the class of arbitrary graphs
that contain a dominating set of size at most $k$. (A dominating set
$S$ in graph is a set of vertices such that every vertex in the graph
is either in $S$ or adjacent to a vertex in $S$.) This class can be
described by the FO sentence $\varphi := \exists x_1 \ldots \exists
x_k \forall y \big( \bigvee_{i = 1}^{i = k} ((y = x_i) \vee E(y, x_i))
\big)$. In any model of $\varphi$, any witnesses to the existential
quantifiers of $\varphi$ form a dominating set, and any such set is a
$k$-crux of the model; then $\varphi$ defines a $k$-hereditary
class. Observe that $\varphi$ is an $\exists^k \forall^*$ sentence,
i.e. a sentence in prenex normal form whose quantifier prefix is a
string of $k$ existential quantifiers followed by universal
quantifiers\footnote{See~\cite{abhisekh-thesis-summary-short} for a
  variety of graph properties of interest in parameterized algorithms
  and finite model theory, that are $k$-hereditary and expressible as
  $\exists^k \forall^*$ sentences.}. By a similar reasoning as above,
it can be shown that any $\exists^k \forall^*$ sentence defines a
$k$-hereditary class. The authors of~\cite{abhisekh-apal} proved that
the converse is true as well, that any FO definable $k$-hereditary
class of arbitrary structures is always definable by an $\exists^k
\forall^*$ sentence, thus proving a \emph{generalized {\lt} theorem},
that we denote $\glt{k}$ (following~\cite{abhisekh-apal}).  Observe that
the {\lt} theorem is a special case of $\glt{k}$ when $k$ is $0$.

The proof of $\glt{k}$ from~\cite{abhisekh-apal} goes via first
showing $\glt{k}$ over a special class of structures called
\emph{$\lambda$-saturated} structures where $\lambda$ is an infinite
cardinal. These structures, intuitively speaking, realize many types
(maximal consistent sets of formulae in a given number of free
variables); in particular, such a structure $\mf{A}$ realizes all the
types that are realized in all structures elementarily equivalent to
$\mf{A}$, i.e. structures which satisfy the same FO sentences as
$\mf{A}$. Then using the fact that every structure has an elementarily
equivalent extension that is $\lambda$-saturated for some $\lambda$,
the truth of $\glt{k}$ is ``transferred'' to all structures. To show
$\glt{k}$ over $\lambda$-saturated structures, a notion dual to
$k$-hereditariness is introduced, called \emph{preservation under
  $k$-ary covered extensions}, that we call \emph{$k$-extension
  closure} in this paper. Given a structure $\mf{A}$, define a set
$\mc{R}$ of substructures of $\mf{A}$ to be a \emph{$k$-ary cover of
  $\mf{A}$} if every set of $k$ elements of $\mf{A}$ is contained in
some structure of $\mc{R}$. We then say $\mf{A}$ is a $k$-extension of
$\mc{R}$. A class is \emph{$k$-extension closed} if every
$k$-extension of a set of structures of the class, is also in the
class. One sees that a class is $k$-extension closed if, and only if,
its complement is $k$-hereditary. Then $\glt{k}$ is shown by proving
its dual form that characterizes $k$-extension closure in terms of
$\forall^k \exists^*$ sentences. The heart of this proof -- Lemma 4.2
of~\cite{abhisekh-apal} -- shows that if $\Gamma$ is the theory of the
$\forall^k \exists^*$ implications of a sentence $\varphi$ that
defines a $k$-extension closed class, then every $\lambda$-saturated
model of $\Gamma$ has a $k$-ary cover consisting of the models of
$\varphi$. It follows that the $\lambda$-saturated model then itself
models $\varphi$, showing that $\varphi$ and $\Gamma$ are equivalent;
then one application of Compactness theorem shows $\varphi$ to be
equivalent to a single sentence of $\Gamma$.

The first result of this paper is motivated by the above proof
of~\cite{abhisekh-apal}. In particular, we give a new proof of
$\glt{k}$ that completely avoids using $\lambda$-saturated structures,
by making the key observation that the full power of
$\lambda$-saturation is hardly used in the proof of the mentioned
Lemma 4.2 of~\cite{abhisekh-apal}.  The formulae that play a central
role in the proof are not arbitrary FO formulae, but are in fact
formulae that have only one quantifier alternation at best.  We
therefore construct just the ``required saturation'' as is needed for
our proof, by showing a ``weaker'' version of the mentioned Lemma 4.2,
that states that for $\Gamma$ and $\varphi$ as above, every model of
$\Gamma$ has an elementarily equivalent extension that might not be
$\lambda$-saturated for any $\lambda$, but still contains a $k$-ary
cover consisting of models of $\varphi$; see
($\ref{lemma:key-lemma-cond-1}) \rightarrow
(\ref{lemma:key-lemma-cond-3}$) of
Lemma~\ref{lemma:elem-ext-containing-k-ary-cover} of this paper. Then
showing (the dual form of) $\glt{k}$ over the class of the mentioned
elementary extensions is sufficient to transfer $\glt{k}$ out to all
structures. The aforementioned implication is in turn shown by
defining in the natural way, the more general notion of a $k$-ary
cover of a structure \emph{in} a superstructure of it, and then using
(transfinite) induction over the $k$-tuples of elements of a given
model $\mf{A}$ of $\Gamma$, to construct an elementary extension
$\mf{A}'$ of $\mf{A}$ such that $\mf{A}$ has a $k$-ary cover
consisting of models of $\varphi$ \emph{in} $\mf{A}'$; see
($\ref{lemma:key-lemma-cond-1}) \rightarrow
(\ref{lemma:key-lemma-cond-2}$) of
Lemma~\ref{lemma:elem-ext-containing-k-ary-cover}.  Applying this
implication iteratively to the elementary extensions it gives, we get
a chain of structures whose union is an elementary extension of
$\mf{A}$ that has a (self-contained) $k$-ary cover of models of
$\varphi$; see ($\ref{lemma:key-lemma-cond-2}) \rightarrow
(\ref{lemma:key-lemma-cond-3}$) of
Lemma~\ref{lemma:elem-ext-containing-k-ary-cover}.  Our new proof is
therefore much ``from the scratch'' as opposed to the proof
in~\cite{abhisekh-apal} which uses established notions of model
theory.

The second result of this paper is a strengthening of the failure of
the {\lt} theorem in the finite. In the research programme of
investigating classical model theoretic results over all finite
structures, that was actively pursued from the mid 80s to mid 90s in
the context of finite model theory~\cite{gurevich-alechina}, one of
the first results identified to fail was the {\lt} theorem. (In fact,
Tait had already shown this failure of the {\lt} theorem in
1959~\cite{tait}.)  Specifically, there is an FO sentence that is
hereditary over the class of all finite structures, but that is not
equivalent over this class to any universal sentence. In the spirit
of~\cite{gurevich-alechina}, one can ask if there is a different
syntactic characterization of hereditariness in the finite, or even a
syntactic (proper) subfragment of FO that is expressive enough to
contain (upto equivalence) all FO sentences that are hereditary when
restricted to the finite. We show in
Theorem~\ref{prop:strong-failure-of-LT-in-the-finite} that for no
fixed $k$, is the class of $\exists^k \forall^*$ sentences such a
subfragment: there is a sentence $\varphi_k$ whose models (even
including infinite ones) form a hereditary class, and yet $\varphi_k$
is not equivalent over all finite structures to any $\exists^k
\forall^*$ sentence.

This result also strengthens the failure of $\glt{k}$ in the finite as
shown in~\cite{wollic12-paper}. For every $k$, the authors
of~\cite{wollic12-paper} constructed a counterexample to $\glt{k}$
(over all finite structures) that is $k$-hereditary but not
$(k-1)$-hereditary. The sentence $\varphi_k$ given by our
Theorem~\ref{prop:strong-failure-of-LT-in-the-finite} provides a
counterexample to $\glt{k}$, that is $l$-hereditary for all $l$.  The
proof of Theorem~\ref{prop:strong-failure-of-LT-in-the-finite} is
based on the Ehrenfeucht-Fr\"aiss\'e (EF) method for showing
inexpressibility results in FO. We construct for each $n$, a model
$\mf{A}_n$ and a non-model $\mf{B}_n$ of $\varphi_k$, such that every
$\exists^k \forall^*$ sentence that is true in $\mf{A}_n$ is also true
in $\mf{B}_n$. The latter is shown by essentially showing that the
Duplicator has a winning strategy in a version of the EF game adapted
to $\exists^k \forall^*$ sentences, in which the Spoiler picks up $k$
elements from $\mf{A}_n$ in the first move, and $n$ elements from
$\mf{B}_n$ in the next move. Interestingly, the sentence $\varphi_k$
itself turns out to be equivalent to an $\exists^{k+1} \forall^*$
sentence.

\tbf{Paper Organization:} In Section~\ref{section:background}, we
introduce terminology and notation used in the paper, and formally
state $\glt{k}$. In Section~\ref{section:glt-k-new-proof}, we present
our new proof of $\glt{k}$ and in
Section~\ref{section:los-tarski-failure}, we prove the strengthened
failure of the {\lt} theorem in the finite. We conclude in
Section~\ref{section:conclusion} by presenting two problems for future
investigation, one concerning the {\lt} theorem and the other
concerning $\glt{k}$, both in the context of all finite structures.

%% file: background.tex
\section{Preliminaries and background}\label{section:background}

We assume the reader is familiar with standard notation and
terminology used in the syntax and semantics of
\text{FO}~\cite{chang-keisler-short}. A \ti{vocabulary} $\tau$ is a
set of predicate, function and constant symbols.  In this paper, we
will always be concerned with arbitrary \emph{finite} vocabularies,
unless explicitly stated otherwise. We denote by $\text{FO}(\tau)$ the
set of all FO formulae over vocabulary $\tau$. A sequence $(x_1,
\ldots, x_k)$ of variables is denoted by $\bar{x}$. A formula $\psi$
whose free variables are among $\bar{x}$, is denoted by
$\psi(\bar{x})$.  A formula with no free variables is called a
\emph{sentence}. An \emph{$\fo(\tau)$ theory} is a set of $\fo(\tau)$
sentences. An $\fo(\tau)$ theory \emph{with free variables $\bar{x}$}
is a set of $\fo(\tau)$ formulae, all of whose free variables are
among $\bar{x}$. When $\tau$ is clear from context, we call an
$\fo(\tau)$ theory, a theory simply. We denote by $\mathbb{N}$, the
natural numbers \emph{including zero}.  We abbreviate a block of
quantifiers of the form $Q x_1 \ldots Q x_k$ by $Q^k \bar{x}$ or $Q
\bar{x}$ (depending on what is better suited for understanding), where
$Q \in \{\forall, \exists\}$ and $k \in \mathbb{N}$.  By $Q^*$, we
mean a block of $k$ $Q$ quantifiers, for some $k \in \mathbb{N}$.  For
every non-zero $n \in \mathbb{N}$, we denote by $\Sigma^0_n$ and
$\Pi^0_n$, the classes of all {\fo} sentences in prenex normal form,
whose quantifier prefixes begin with $\exists$ and $\forall$
respectively, and consist of $n-1$ alternations of quantifiers. We
call $\Sigma^0_1$ formulae \emph{existential} and $\Pi^0_1$ formulae
\emph{universal}.  We call $\Sigma^0_2$ formulae with $k$ existential
quantifiers \emph{$\exists^k \forall^*$ formulae}, and $\Pi^0_2$
formulae with $k$ universal quantifiers \emph{$\forall^k \exists^*$
  formulae}.

We use standard notions of $\tau$-structures (denoted $\mf{A}, \mf{B}$
etc.; we refer to these simply as structures when $\tau$ is clear from
context), substructures (denoted $\mf{A} \subseteq \mf{B}$),
extensions, isomorphisms (denoted $\mf{A} \cong \mf{B}$), isomorphic
embeddings (denoted $\mf{A} \hookrightarrow \mf{B}$), elementary
equivalence (denoted $\mf{A} \equiv \mf{B}$), elementary substructures
(denoted $\mf{A} \preceq \mf{B}$) and elementary extensions, as
defined in ~\cite{chang-keisler-short}. Given a structure
$\mathfrak{A}$, we use $\mathsf{U}_{\mathfrak{A}}$ to denote the
universe of $\mathfrak{A}$, and $|\mf{A}|$ to denote the size (or
\ti{power}) of $\mf{A}$ which is the cardinality of
$\mathsf{U}_{\mathfrak{A}}$.  For an {\fo} sentence $\varphi$ and an
   {\fo} theory $T$, we denote by $\mathfrak{A} \models \varphi$ and
   $\mf{A} \models T$ that $\mf{A}$ is a model of $\varphi$ and $T$
   respectively. In Section~\ref{section:glt-k-new-proof} of the
   paper, we consider structures that could be finite or infinite,
   whereas in Section~\ref{section:los-tarski-failure} we restrict
   ourselves to only finite structures.

Finally, we use standard abbreviations of English phrases that
commonly appear in mathematical literature. Specifically, `w.l.o.g'
stands for `without loss of generality', `iff' stands for `if and only
if', and `resp.'  stands for `respectively'.

\subsection{The generalized {\lt} theorem}

We recall the notions of preservation under substructures modulo
$k$-cruxes, $k$-ary covered extensions and preservation under $k$-ary
covered extensions introduced in~\cite{abhisekh-apal}, that we
resp. call in this paper $k$-hereditariness, $k$-extensions and
$k$-extension closure. These notions for $k = 0$ correspond exactly to
hereditariness, extensions and extension closure resp.

\begin{definition}[Definition 3.1,~\cite{abhisekh-apal}]\label{defn:PSC(k)}\leavevmode
  \begin{enumerate}[nosep]
    \item[a.] Let $\cl{U}$ be a class of arbitrary structures and $k
      \in \mathbb{N}$. A subclass $\cl{S}$ of $\cl{U}$ is said to be
      \emph{$k$-hereditary over $\cl{U}$}, if for every structure
      $\mf{A}$ of $\cl{S}$, there is a set $C \subseteq
      \mathsf{U}_{\mf{A}}$ of size $\leq k$ such that if $\mf{B}
      \subseteq \mf{A}$, $\mf{B}$ contains $C$ and $\mf{B} \in
      \cl{U}$, then $\mf{B} \in \cl{S}$. The set $C$ is called a
      \emph{$k$-crux} of $\mf{A}$ w.r.t. $\cl{S}$ over $\cl{U}$.
    \item[b.] Given theories $T$ and $V$, we say $T$ is
      \emph{$k$-hereditary modulo $V$}, if the class of models of $T
      \cup V$ is $k$-hereditary over the class of models of $V$. A
      sentence $\varphi$ is $k$-hereditary modulo $V$ if the theory
      $\{\varphi\}$ is $k$-hereditary modulo $V$.
  \end{enumerate}
\end{definition}

\begin{definition}[Definitions 3.5 and 3.8,~\cite{abhisekh-apal}]\label{defn:PCE(k)}\leavevmode
  \begin{enumerate}[nosep]
    \item[a.] Given a structure $\mf{A}$, a non-empty collection
      $\mc{R}$ of substructures of $\mf{A}$ is said to be a
      \emph{$k$-ary cover} of $\mf{A}$ if for every set $C \subseteq
      \mathsf{U}_{\mf{A}}$ of size $\leq k$, there is a structure in
      $\mc{R}$ that contains $C$.  We call $\mf{A}$ a
      \emph{$k$-extension} of $\mc{R}$.

    \item[b.] For a class $\cl{U}$ of arbitrary structures and $k
      \in \mathbb{N}$, a subclass $\cl{S}$ of $\cl{U}$ is said to be
      \emph{$k$-extension closed over $\cl{U}$} if for every
      collection $\mc{R}$ of structures of $\cl{S}$, if $\mf{A}$ is a
      $k$-extension of $\mc{R}$ and $\mf{A} \in \cl{U}$, then $\mf{A}
      \in \cl{S}$.
    \item[c.] Given theories $V$ and $T$, we say $T$ is
      \emph{$k$-extension closed modulo $V$} if the class of models of
      $T \cup V$ is $k$-extension closed over the class of models of
      $V$. A sentence $\varphi$ is $k$-extension closed modulo $V$ if
      the theory $\{\varphi\}$ is $k$-extension closed modulo $V$.
  \end{enumerate}
\end{definition}

We extend the above definitions slightly to formulae and theories with
free variables. Given a vocabulary $\tau$, let $\tau_n$ denote the
vocabulary obtained by expanding $\tau$ with $n$ fresh and distinct
constant symbols $c_1, \ldots, c_n$. For a given $\fo(\tau)$ theory
$T(x_1, \ldots, x_n)$, let $T'$ denote the $\fo(\tau_n)$ theory
(without free variables) obtained by substituting $c_i$ for $x_i$ in
$T(x_1, \ldots, x_n)$ for each $i \in \{1, \ldots, n\}$. Then we say
$T(x_1, \ldots, x_n)$ is $k$-hereditary, resp. $k$-extension closed,
modulo a theory $V$ if $T'$ is $k$-hereditary, resp. $k$-extension
closed, modulo $V$ where $V$ is seen as an $\fo(\tau_n)$ theory. A
formula $\varphi(x_1, \ldots, x_n)$ is $k$-hereditary,
resp. $k$-extension closed, modulo $V$ if the theory $\{\varphi(x_1,
\ldots, x_n)\}$ is $k$-hereditary, resp. $k$-extension closed, modulo
$V$.  The following lemma establishes the duality of the introduced
preservation properties.

\begin{lemma}[Lemma 3.9,~\cite{abhisekh-apal}]\label{lemma:duality} 
  Let $\cl{U}$ be a class of arbitrary structures, $\cl{S}$ be a
  subclass of $\cl{U}$ and $\overline{\cl{S}}$ be the complement of
  $\cl{S}$ in $\cl{U}$. Then $\cl{S}$ is $k$-hereditary over $\cl{U}$
  iff $\overline{\cl{S}}$ is $k$-extension closed over $\cl{U}$, for
  each $k \in \mathbb{N}$. In particular, if $\cl{U}$ is defined by a
  theory $V$, then a formula $\varphi(\bar{x})$ is $k$-hereditary
  modulo $V$ iff $\neg \varphi(\bar{x})$ is $k$-extension closed
  modulo $V$.
\end{lemma}

We now recall $\glt{k}$ as proved in~\cite{abhisekh-apal}. This
theorem gives syntactic characterizations of FO definable
$k$-hereditary and $k$-extension closed classes of structures. Observe
that the case of $k = 0$ gives exactly the {\lt} theorem. Below, for
$\fo(\tau)$ formulae $\varphi(\bar{x})$ and $\psi(\bar{x})$ where
$\bar{x} = (x_1, \ldots, x_n)$, we say $\varphi(\bar{x})$ is
equivalent to $\psi(\bar{x})$ \emph{modulo $V$} if for every
$\tau$-structure $\mf{A}$ and every $n$-tuple $\bar{a}$ from $\mf{A}$,
we have $(\mf{A}, \bar{a})$ is a model of $\{\varphi(\bar{x})\} \cup
V$ iff it is a model of $\{\psi(\bar{x})\} \cup V$.

\begin{theorem}[Generalized {\lt} theorem: $\glt{k}$; Corollaries 4.4 and 4.6, ~\cite{abhisekh-apal}]\label{theorem:glt-k}
  Let $\varphi(\bar{x})$ and $V$ be a given formula and theory
  respectively, and $k \in \mathbb{N}$. Then the following are true:
  \begin{enumerate}[nosep]
    \item The formula $\varphi(\bar{x})$ is $k$-hereditary modulo $V$
      iff it is equivalent modulo $V$ to an $\exists^k \forall^*$
      formula whose free variables are among
      $\bar{x}$.\label{theorem:glt-k-subs}
    \item The formula $\varphi(\bar{x})$ is $k$-extension closed
      modulo $V$ iff it is equivalent modulo $V$ to a $\forall^k
      \exists^*$ formula whose free variables are among
      $\bar{x}$.\label{theorem:glt-k-ext}
  \end{enumerate}
\end{theorem}

%% file: proof-of-GLT-using-ascending-chains-2.0.tex
\section{A new proof of $\glt{k}$}\label{section:glt-k-new-proof}

We give a new proof to a more general result than
Theorem~\ref{theorem:glt-k}, from~\cite{abhisekh-apal}. This result is
a generalization of the ``extensional'' version of $\glt{k}$ to
theories.  We extend in the natural way the aforestated notion of
equivalence modulo a theory, of formulae, to theories with free
variables.

\begin{theorem}\label{theorem:glt-k-extensional-version-for-theories}
A theory $T(\bar{x})$ is $k$-extension closed modulo a theory $V$ if,
and only if, $T(\bar{x})$ is equivalent modulo $V$ to a theory
(consisting) of $\forall^k \exists^*$ formulae all of whose free
variables are among $\bar{x}$.
\end{theorem}
Using the above result, Theorem~\ref{theorem:glt-k} can be proved as
below.

\begin{proof}[Proof of Theorem~\ref{theorem:glt-k}]
We prove part (\ref{theorem:glt-k-ext}) of
Theorem~\ref{theorem:glt-k}. Part (\ref{theorem:glt-k-subs}) of
Theorem~\ref{theorem:glt-k} easily follows from part
(\ref{theorem:glt-k-ext}) and Lemma~\ref{lemma:duality}.

The `If' direction is straightforward. Let $\varphi(\bar{x})$ be
equivalent modulo $V$ to the $\forall^k \exists^*$ formula
$\psi(\bar{x})$. Then the theory $\{\varphi(\bar{x})\}$ is equivalent
modulo $V$ to the theory $\{\psi(\bar{x})\}$. Then
$\{\varphi(\bar{x})\}$, and hence $\varphi(\bar{x})$, is $k$-extension
closed modulo $V$ by
Theorem~\ref{theorem:glt-k-extensional-version-for-theories}.  For the
`Only if' direction, let $\varphi(\bar{x})$ be $k$-extension closed
modulo $V$; then so is the theory $\{\varphi(\bar{x})\}$. By
Theorem~\ref{theorem:glt-k-extensional-version-for-theories},
$\{\varphi(\bar{x})\}$ is equivalent to a theory $Z(\bar{x})$ of
$\forall^k \exists^*$ formulae whose free variables are among
$\bar{x}$. By Compactness theorem, $\{\varphi(\bar{x})\}$ is
equivalent modulo $V$ to a finite subset $Y(\bar{x})$ of
$Z(\bar{x})$. Then $\varphi(\bar{x})$ is equivalent modulo $V$ to the
conjunction of the formulae of $Y(\bar{x})$. Since any conjunction of
$\forall^k \exists^*$ formulae is equivalent (modulo any theory) to a
single $\forall^k \exists^*$ formula, the result follows.
\end{proof}

Towards Theorem~\ref{theorem:glt-k-extensional-version-for-theories},
we first recall some important notions and results from the classical
model theory literature~\cite{chang-keisler-short} that are needed for our
proof.

\begin{lemma}[Corollary 5.4.2, Chapter 5~\cite{chang-keisler-short}]\label{lemma:exis-amalgam}
  Let $\mf{A}$ and $\mf{B}$ be structures such that every existential
  sentence that is true in $\mf{B}$ is true in $\mf{A}$. Then $\mf{B}$
  is isomorphically embeddable in an elementary extension of $\mf{A}$.
\end{lemma}

Given a cardinal $\lambda$, an \emph{ascending chain}, or simply a
chain, $(\mf{A}_\eta)_{\eta < \lambda}$ of structures is a sequence
$\mf{A}_0, \mf{A}_1, \ldots$ of structures such that $\mf{A}_{0}
\subseteq \mf{A}_{1} \subseteq \ldots$. The \emph{union} of this chain
is a structure $\mf{A}$ defined as follows: (i) $\univ{\mf{A}} =
\bigcup_{\eta < \lambda} \univ{\mf{A}_\eta}$, (ii) $c^{\mf{A}} =
c^{\mf{A}_\eta}$ for every constant symbol $c \in \tau$ and every
$\eta < \lambda$, (iii) $R^{\mf{A}} = \bigcup_{\eta < \lambda}
R^{\mf{A}_\eta}$ for every relation symbol $R \in \tau$, and (iv)
$f^{\mf{A}} = \bigcup_{\eta < \lambda} f^{\mf{A}_\eta}$ for every
function symbol $f \in \tau$ (here, in taking the union of functions,
we view an $n$-ary function as its corresponding $(n+1)$-ary
relation). Observe that $\mf{A}$ is well-defined. We denote $\mf{A}$
as $\bigcup_{\eta < \lambda} \mf{A}_\eta$. If it is additionally the
case that $\mf{A}_{0} \preceq \mf{A}_{1} \preceq \ldots$ above, then
we say $(\mf{A}_\eta)_{\eta < \lambda}$ is an \emph{elementary
  chain}. We now have the following result.

\begin{theorem}[Tarski-Vaught elementary chain theorem, Theorem 3.1.9, Chapter 3~\cite{chang-keisler-short}]\label{theorem:elem-chain-theorem} Let $(\mf{A}_\eta)_{\eta < \lambda}$ be an
  elementary chain of structures. Then $\bigcup_{\eta < \lambda}
  \mf{A}_\eta$ is an elementary extension of $\mf{A}_\eta$ for each
  $\eta < \lambda$.
\end{theorem}

The key element of our proof of
Theorem~\ref{theorem:glt-k-extensional-version-for-theories} is the
notion of a $k$-ary cover of a structure $\mf{A}$ \emph{in an
  extension of $\mf{A}$}. Below is the definition. Observe that this
notion generalizes the notion of $k$-ary cover seen in
Definition~\ref{defn:PCE(k)} -- the latter corresponds to the notion
in Definition~\ref{defn:k-ary-cover-inside-an-elem-ext}, with
$\mf{A}^+$ being the same as $\mf{A}$.

\begin{definition}\label{defn:k-ary-cover-inside-an-elem-ext}
Let $\mf{A}$ be a structure and $\mf{A}^+$ be an extension of
$\mf{A}$. A non-empty collection $\mc{R}$ of substructures of
$\mf{A}^+$ is said to be a \emph{$k$-ary cover of $\mf{A}$ in
  $\mf{A}^+$} if for every $k$-tuple $\bar{a}$ of elements of
$\mf{A}$, there exists a structure in $\mc{R}$ that contains (the
elements of) $\bar{a}$.
\end{definition}

The following lemma is at the heart of our proof. It (along with its
application in proving
Theorem~\ref{theorem:glt-k-extensional-version-for-theories}) shows
why ``full'' $\lambda$-saturation as is used in a similar result
(Lemma 4.2) in~\cite{abhisekh-apal}, is not needed for
Theorem~\ref{theorem:glt-k-extensional-version-for-theories}.  Below,
$\Gamma = \{ \varphi \mid (V \cup T) \rightarrow
\varphi~\text{where}~\varphi~\text{is a}~\forall^k
\exists^*~\text{sentence}\}$. Also a \emph{consistent} theory is one
that has a model.

%% Lemma notes:
%% 1. Observe that \mf{A} being a model of V is required!

\begin{lemma}\label{lemma:elem-ext-containing-k-ary-cover}
Let $V$ and $T$ be consistent theories and $k \in \mathbb{N}$. Let
$\Gamma$ be the set of $\forall^k \exists^*$ consequences of $T$
modulo $V$. Then for every model $\mf{A}$ of $V$, the following are
equivalent:
\begin{enumerate}[nosep]
\item \vspace{2pt} $\mf{A}$ is a model of $V \cup
  \Gamma$.\label{lemma:key-lemma-cond-1}
\item \vspace{2pt} There exists an elementary extension $\mf{A}^+$ of
  $\mf{A}$ and a $k$-ary cover $\mc{R}$ of $\mf{A}$ in $\mf{A}^+$ such
  that $\mf{B} \models (V \cup T)$ for every $\mf{B} \in
  \mc{R}$.\label{lemma:key-lemma-cond-2}
\item \vspace{2pt} There exists an elementary extension $\mf{A}^+$ of
  $\mf{A}$ and a $k$-ary cover $\mc{R}$ of $\mf{A}^+$ (in $\mf{A}^+$)
  such that $\mf{B} \models (V \cup T)$ for every $\mf{B} \in
  \mc{R}$.\label{lemma:key-lemma-cond-3}
\end{enumerate}
\end{lemma}

Using the above lemma,
Theorem~\ref{theorem:glt-k-extensional-version-for-theories} can be
proved as follows.

\begin{proof}[Proof of Theorem~\ref{theorem:glt-k-extensional-version-for-theories}]
  We prove the theorem for theories without free variables; the proof
  for theories with free variables follows from definitions.
  
  \und{If:} Suppose $T$ is equivalent modulo $V$ to a theory $Z$ of
  $\forall^k \exists^*$ sentences. Let $\mf{A} \models V$ and let
  $\mc{R}$ be a $k$-ary cover of $\mf{A}$ consisting of models of $V
  \cup T$. We show that $\mf{A} \models T$. Consider a sentence
  $\varphi := \forall^k \bar{x} \psi(\bar{x}) \in Z$ where
  $\psi(\bar{x})$ is an existential formula. Let $\bar{a}$ be a
  $k$-tuple from $\mf{A}$. Since $\mc{R}$ is a $k$-ary cover of
  $\mf{A}$, there exists $\mf{B} \in \mc{R}$ such that $\mf{B}
  \subseteq \mf{A}$ and $\mf{B}$ contains $\bar{a}$. Since $\mf{B}
  \models V \cup T$, we have $\mf{B} \models Z$ (since $Z$ and $T$ are
  equivalent modulo $V$); then $\mf{B} \models \varphi$ and hence
  $(\mf{B}, \bar{a}) \models \psi(\bar{x})$. Since existential
  formulae are preserved under extensions by {\lt} theorem, we have
  $(\mf{A}, \bar{a}) \models \psi(\bar{x})$. Since $\bar{a}$ is an
  arbitrary $k$-tuple of $\mf{A}$, we have $\mf{A} \models
  \varphi$. Finally, since $\varphi$ is an arbitrary sentence of $Z$,
  we have $\mf{A} \models Z$, and hence $\mf{A} \models T$.

  \und{Only if:} Conversely, suppose $T$ is $k$-extension closed
  modulo $V$. If $V \cup T$ is unsatisfiable, we are trivially
  done. Else, let $\Gamma$ be the set of $\forall^k \exists^*$
  consequences of $T$ modulo $V$.  Then $(V \cup T) \rightarrow (V
  \cup \Gamma)$. Conversely, suppose $\mf{A} \models (V \cup
  \Gamma)$. By Lemma~\ref{lemma:elem-ext-containing-k-ary-cover},
  there exists an elementary extension $\mf{A}^+$ of $\mf{A}$ (hence
  $\mf{A}^+ \models V$) for which there is a $k$-ary cover consisting
  of models of $V \cup T$. Then $\mf{A}^+ \models T$ since $T$ is
  $k$-extension closed modulo $V$, whereby $\mf{A} \models T$.  In
  other words, $(V \cup \Gamma) \rightarrow (V \cup T)$, so that $T$
  is equivalent to $\Gamma$ modulo V. Then $\Gamma$ is the desired
  $\forall^k \exists^*$ theory.
\end{proof}
  
Towards the proof of
Lemma~\ref{lemma:elem-ext-containing-k-ary-cover}, we would require
the following auxiliary lemma.

\begin{lemma}\label{lemma:auxiliary}
  Let $V, T$ and $\Gamma$ be as in the statement of
  Lemma~\ref{lemma:elem-ext-containing-k-ary-cover}, and suppose
  $\mf{A} \models (V \cup \Gamma)$. Given an elementary extension
  $\mf{A}'$ of $\mf{A}$ and a $k$-tuple $\bar{a}$ of $\mf{A}$, there
  exist an elementary extension $\mf{A}''$ of $\mf{A}'$ and a
  substructure $\mf{B}$ of $\mf{A}''$ such that (i) $\mf{B}$ contains
  $\bar{a}$ and (ii) $\mf{B} \models (V \cup T)$.
\end{lemma}
\begin{proof}
  Let $\mathsf{tp}_{\Pi, \mf{A}, \bar{a}}(\bar{x})$ denote the
  $\Pi^0_1$-type of $\bar{a}$ in $\mf{A}$, that is, the set of all
  $\Pi^0_1$ formulae that are true of $\bar{a}$ in $\mf{A}$ (so
  $|\bar{x}| = |\bar{a}|$).  Let $Z(\bar{x})$ be the theory given by
  $Z(\bar{x}) := V \cup T \cup \mathsf{tp}_{\Pi, \mf{A},
    \bar{a}}(\bar{x})$. We show below that $Z(\bar{x})$ is
  satisfiable. Assuming this, it follows that if $(\mf{D}, \bar{d})
  \models Z(\bar{x})$, then every existential sentence that is true in
  $(\mf{D}, \bar{d})$ is also true in $(\mf{A}, \bar{a})$, and hence
  in $(\mf{A}', \bar{a})$. Then by Lemma~\ref{lemma:exis-amalgam},
  there is an isomorphic embedding $f$ of $(\mf{D}, \bar{d})$ in an
  elementary extension $(\mf{A}'', \bar{a})$ of $(\mf{A}',
  \bar{a})$. If the vocabulary of $\mf{A}$ is $\tau$, then taking
  $\mf{B}$ to be the $\tau$-reduct of the image of $(\mf{D}, \bar{d})$
  under $f$, we see that $\mf{B}$ and $\mf{A}''$ are as desired.

  We show $Z(\bar{x})$ is satisfiable by contradiction. Suppose
  $Z(\bar{x})$ is inconsistent; then by Compactness theorem, there is
  a finite subset of $Z(\bar{x})$ that is inconsistent. Since
  $\mathsf{tp}_{\Pi, \mf{A}, \bar{a}}(\bar{x})$ is closed under finite
  conjunctions and since each of $\mathsf{tp}_{\Pi, \mf{A},
    \bar{a}}(\bar{x})$, $V$ and $T$ is consistent, there exists
  $\psi(\bar{x}) \in \mathsf{tp}_{\Pi, \mf{A}, \bar{a}}(\bar{x})$ such
  that $V \cup T \cup \{\psi(\bar{x})\}$ is inconsistent. In other
  words, $(V \cup T) \rightarrow \neg \psi(\bar{x})$.  Since $V \cup
  T$ has no free variables, we have $(V \cup T) \rightarrow \varphi$,
  where $\varphi := \forall^k \bar{x}$ $\neg \psi(\bar{x})$.  Observe
  that $\neg \psi(\bar{x})$ is equivalent to an existential formula;
  then $\varphi$ is equivalent to a sentence in $\Gamma$, and hence
  $\mf{A} \models \varphi$.  Then $(\mf{A}, \bar{a}) \models \neg
  \psi(\bar{x})$, contradicting our inference that $\psi(\bar{x}) \in
  \mathsf{tp}_{\Pi, \mf{A}, \bar{a}}(\bar{x})$.
\end{proof}

\begin{proof}[Proof of Lemma~\ref{lemma:elem-ext-containing-k-ary-cover}]
  \und{$(\ref{lemma:key-lemma-cond-3}) \rightarrow
    (\ref{lemma:key-lemma-cond-1})$:} This implication is established
  along similar lines as the `If' direction of
  Theorem~\ref{theorem:glt-k-extensional-version-for-theories}. Let
  $\varphi := \forall^k \bar{x} \psi(\bar{x})$ for a $\Sigma^0_1$
  formula $\psi(\bar{x})$, be a sentence of $\Gamma$, and let
  $\bar{a}$ be a $k$-tuple of $\mf{A}$. Since $\mc{R}$ is a $k$-ary
  cover of $\mf{A}^+$, there exists $\mf{B}_{\bar{a}} \in \mc{R}$ such
  that $\mf{B}_{\bar{a}}$ contains $\bar{a}$ and $\mf{B}_{\bar{a}}
  \models (V \cup T)$. Then $\mf{B}_{\bar{a}} \models \Gamma$ so that
  $\mf{B}_{\bar{a}} \models \varphi$, and hence $(\mf{B}_{\bar{a}},
  \bar{a}) \models \psi(\bar{x})$. Since $\psi(\bar{x})$ is a
  $\Sigma^0_1$ formula and $\mf{B}_{\bar{a}} \subseteq \mf{A}^+$, we
  have $(\mf{A}^+, \bar{a}) \models \psi(\bar{x})$ since existential
  formulae are preserved under extensions by the {\lt} theorem. Then
  $(\mf{A}, \bar{a}) \models \psi(\bar{x})$ since $\mf{A} \preceq
  \mf{A}^+$. Since $\bar{a}$ is an arbitrary $k$-tuple of $\mf{A}$, we
  have $\mf{A} \models \varphi$.

  \vspace{2pt} \noindent \und{$(\ref{lemma:key-lemma-cond-1})
    \rightarrow (\ref{lemma:key-lemma-cond-2})$:} We have two cases
  here depending on whether $\mf{A}$ is finite or infinite.

  (1) $\mf{A}$ is finite: Given a $k$-tuple $\bar{a}$ of $\mf{A}$, by
  Lemma~\ref{lemma:auxiliary} there exists an elementary extension
  $\mf{A}''$ of $\mf{A}$ and a substructure $\mf{B}_{\bar{a}}$ of
  $\mf{A}''$ such that (i) $\mf{B}_{\bar{a}}$ contains $\bar{a}$ and
  (ii) $\mf{B}_{\bar{a}} \models (V \cup T)$. Since $\mf{A}$ is
  finite, and since elementary equivalence is the same as isomorphism
  over finite structures~\cite{chang-keisler-short}, we have $\mf{A}''
  = \mf{A}$. Then taking $\mf{A}^+ = \mf{A}$ and $\mc{R} =
  \{\mf{B}_{\bar{a}} \mid \bar{a} \in (\mathsf{U}_{\mf{A}})^k\}$, we
  see that $\mf{A}^+$ and $\mc{R}$ are respectively indeed the desired
  elementary extension of $\mf{A}$ and $k$-ary cover of $\mf{A}$ in
  $\mf{A}^+$.

  (2) $\mf{A}$ is infinite: The proof for this case is along the lines
  of the proof of the characterization of $\Pi^0_2$ sentences in terms
  of the property of preservation under unions of chains (see proof of
  Theorem 3.2.3 in Chapter 3 of~\cite{chang-keisler-short}). Let
  $\lambda$ be the successor cardinal of $|\mf{A}|$ and
  $(\bar{a}_\kappa)_{\kappa < \lambda}$ be an enumeration of the
  $k$-tuples of $\mf{A}$. For $\eta \leq \lambda$, given sequences
  $(\mf{E}_\kappa)_{\kappa < \eta}$ and $(\mf{F}_\kappa)_{\kappa <
    \eta}$ of structures, we say that $\mc{P}((\mf{E}_\kappa)_{\kappa
    < \eta}, (\mf{F}_\kappa)_{\kappa < \eta})$ is true iff
  $(\mf{E}_\kappa)_{\kappa < \eta}$ is an elementary chain and $\mf{A}
  \preceq \mf{E}_0$, and for each $\kappa < \eta$, we have (i)
  $\mf{F}_\kappa \subseteq \mf{E}_\kappa$, (ii) $\mf{F}_{\kappa}$
  contains $\bar{a}_{\kappa}$, and (iii) $\mf{F}_\kappa \models (V
  \cup T)$. We show below the existence of sequences
  $(\mf{A}_\kappa)_{\kappa < \lambda}$ and $(\mf{B}_\kappa)_{\kappa <
    \lambda}$ of structures such that $\mc{P}((\mf{A}_\kappa)_{\kappa
    < \lambda}, (\mf{B}_\kappa)_{\kappa < \lambda})$ is true. Then
  taking $\mf{A}^+ = \bigcup_{\kappa < \lambda} \mf{A}_\kappa$ and
  $\mc{R} = \{\mf{B}_\kappa \mid \kappa < \lambda\}$, we see by
  Theorem~\ref{theorem:elem-chain-theorem} that $\mf{A}^+$ and
  $\mc{R}$ are respectively indeed the elementary extension of
  $\mf{A}$ and $k$-ary cover of $\mf{A}$ in $\mf{A}^+$ as desired.

  We construct the sequences $(\mf{A}_\kappa)_{\kappa < \lambda}$ and
  $(\mf{B}_\kappa)_{\kappa < \lambda}$ by constructing for each
  positive ordinal $\eta < \lambda$, the partial (initial) sequences
  $(\mf{A}_\kappa)_{\kappa < \eta}$ and $(\mf{B}_\kappa)_{\kappa <
    \eta}$ and showing that $\mc{P}((\mf{A}_\kappa)_{\kappa < \eta},
  (\mf{B}_\kappa)_{\kappa < \eta})$ is true. We do this by
  (transfinite) induction on $\eta$.  For the base case of $\eta = 1$,
  we see by Lemma~\ref{lemma:auxiliary} that if $\mf{A}' = \mf{A}$,
  then there exists an elementary extension $\mf{A}''$ of $\mf{A}$ and
  a substructure $\mf{B}$ of $\mf{A}''$ such that (i) $\mf{B}$
  contains $\bar{a}_0$ and (ii) $\mf{B} \models (V \cup T)$. Then
  taking $\mf{A}_0 = \mf{A}''$ and $\mf{B}_0 = \mf{B}$, we see that
  $\mc{P}((\mf{A}_0), (\mf{B}_0))$ is true.  As the induction
  hypothesis, assume that we have constructed sequences
  $(\mf{A}_\kappa)_{\kappa < \eta}$ and $(\mf{B}_\kappa)_{\kappa <
    \eta}$ such that $\mc{P}((\mf{A}_\kappa)_{\kappa < \eta},
  (\mf{B}_\kappa)_{\kappa < \eta})$ is true.  Then by
  Theorem~\ref{theorem:elem-chain-theorem}, the structure $\mf{A}' =
  \bigcup_{\kappa < \eta} \mf{A}_\kappa$ is such that $\mf{A}_\kappa
  \preceq \mf{A}'$ for each $\kappa < \eta$. Then for the tuple
  $\bar{a}_\eta$ of $\mf{A}$, by Lemma~\ref{lemma:auxiliary}, there
  exists an elementary extension $\mf{C}$ of $\mf{A}'$ and a
  substructure $\mf{D}$ of $\mf{C}$ such that (i) $\mf{D}$ contains
  $\bar{a}_\eta$ and (ii) $\mf{D} \models (V \cup T)$.  Then taking
  $\mf{A}_\eta = \mf{C}$ and $\mf{B}_\eta = \mf{D}$, and letting $\mu$
  be the successor ordinal of $\eta$, we see that
  $\mc{P}((\mf{A}_\kappa)_{\kappa < \mu}, (\mf{B}_\kappa)_{\kappa <
    \mu})$ is true, completing the induction.

  \vspace{2pt} \noindent \und{$(\ref{lemma:key-lemma-cond-2})
    \rightarrow (\ref{lemma:key-lemma-cond-3})$:} Applying
  (\ref{lemma:key-lemma-cond-2}) to $\mf{A}$ and then iteratively to
  the elementary extensions that (\ref{lemma:key-lemma-cond-2})
  produces, we get a sequence $(\mf{A}_i)_{i \ge 0}$ of structures and
  a sequence $(\mc{R}_i)_{i \ge 0}$ of collections of structures with
  the following properties:
  \begin{enumerate}[nosep]
    \item \vspace{2pt} $(\mf{A}_i)_{i \ge 0}$ is an elementary chain
      such that $\mf{A}_0 = \mf{A}$ (whereby $\mf{A}_i \models V$ for
      $i \ge 0$).
    \item \vspace{2pt} For each $i \ge 0$, $\mc{R}_i$ is a $k$-ary
      cover of $\mf{A}_i$ in $\mf{A}_{i+1}$ such that $\mf{B} \models
      (V \cup T)$ for every $\mf{B} \in \mc{R}_i$.
  \end{enumerate}

  Consider the structure $\mf{A}^+ = \bigcup_{i \ge 0} \mf{A}_i$. By
  Theorem~\ref{theorem:elem-chain-theorem}, we have $\mf{A}_i \preceq
  \mf{A}^+$ for each $i \ge 0$, and (hence) that $\mf{A}^+ \models
  V$. Consider any $k$-tuple $\bar{a}$ of $\mf{A}^+$; there exists $j
  \ge 0$ such $\bar{a}$ is contained in $\mf{A}_j$. Then there exists
  a structure $\mf{B}_{\bar{a}} \in \mc{R}_j$ such that (i)
  $\mf{B}_{\bar{a}}$ contains $\bar{a}$ and (ii) $\mf{B}_{\bar{a}}
  \models (V \cup T)$. Since $\mf{B}_{\bar{a}} \in \mc{R}_j$, we have
  $\mf{B}_{\bar{a}} \subseteq \mf{A}_{j+1}$ and since $\mf{A}_{j+1}
  \preceq \mf{A}^+$, we have $\mf{B}_{\bar{a}} \subseteq
  \mf{A}^+$. Then $\mc{R} = \{\mf{B}_{\bar{a}} \mid \bar{a}~\text{is
    a}~k\text{-tuple from}~\mf{A}^+\}$ is the desired $k$-ary cover of
  $\mf{A}^+$ such that $\mf{B} \models (V \cup T)$ for each $\mf{B}
  \in \mc{R}$.
\end{proof}

%% file: strong-failure-of-LT-in-finite.tex
\section{A stronger failure of {\lt} theorem in the finite}\label{section:los-tarski-failure}

In this section, we strengthen the known failure of the {\lt} theorem
in the finite~\cite{gurevich-alechina}. As a consequence, we get a
strengthening of the failure of $\glt{k}$ in the finite for each $k$,
over the one proved in~\cite{wollic12-paper}. Below, by $\varphi_k$ is
($k$-)hereditary over $\cl{S}$ we mean that the class of finite models
of $\varphi_k$ is ($k$-)hereditary over $\cl{S}$.

\begin{theorem}\label{prop:strong-failure-of-LT-in-the-finite}
There exists a vocabulary $\tau$ such that if $\cl{S}$ is the class of
all finite $\tau$-structures, then for each $k \ge 0$, there exists an
FO($\tau$) sentence $\varphi_k$ that is hereditary over $\cl{S}$, but
that is not equivalent over $\cl{S}$, to any $\exists^k \forall^*$
sentence. It follows that there is a sentence that is $k$-hereditary
over $\cl{S}$ ($\varphi_k$ being one such sentence) but that is not
equivalent over $\cl{S}$ to any $\exists^k \forall^*$ sentence.
\end{theorem}

\begin{proof}
The second part of the theorem follows from the first part since a
sentence that is hereditary over $\cl{S}$ is also $k$-hereditary over
$\cl{S}$ for each $k \ge 0$. We now prove the first part of the
theorem.  Consider the vocabulary $\tau = \{\leq, S, P, c, d\}$ where
$\leq$ and $S$ are binary relation symbols, $P$ is a unary relation
symbol, and $c$ and $d$ are constant symbols. The sentence $\varphi_k$
is constructed along the lines of the counterxample to the {\lt}
theorem in the finite as given in~\cite{gurevich-alechina}.
\[
\begin{array}{lll}
\varphi_k & := & (\xi_1 \wedge \xi_2 \wedge \xi_3) \wedge \neg (\xi_4 \wedge \xi_5) \\
\xi_1 & := & \text{``}\leq~\text{is a linear order~''} \\
\xi_2 & := & \text{``}~c~\text{is minimum under}~\leq~\text{and}~d~\text{is maximum under}~\leq~\text{"}\\
\xi_3 & := & \forall x \forall y ~S(x, y) \rightarrow \text{``}~y~\text{is the successor of}~x~\text{under}~\leq~\text{"}\\
\xi_4 & := & \forall x ~(x \neq d) \rightarrow \exists y S(x, y)\\
\xi_5 & := & \text{``~There exist at most}~k~\text{elements in (the set interpreting)}~P~\text{"}\\
\end{array}
\]

Each of $\xi_1, \xi_2, \xi_3$ and $\xi_5$ can be expressed using a
universal sentence. In particular, $\xi_1$ and $\xi_3$ can be
expressed using a $\forall^3$ sentence each, $\xi_2$ using a $\forall$
sentence, and $\xi_5$ using a $\forall^{k+1}$ sentence. Then
$\varphi_{k}$ is equivalent to an $\exists^{k+1}\forall^3$ sentence.

We first show that $\varphi_k$ is hereditary over $\cl{S}$, by showing
that $\psi_k := \neg \varphi_k$ is extension closed over $\cl{S}$. Let
$\mf{A} \models \psi_k$ and $\mf{A} \subseteq \mf{B}$. If $\alpha :=
(\xi_1 \,\wedge\, \xi_2 \wedge\, \xi_3)$ is such that $\mf{A} \models
\neg \alpha$, then since $\neg \alpha$ is equivalent to an existential
sentence, we have $\mf{B} \models \neg \alpha$; then $\mf{B} \models
\psi_k$. Else, $\mf{A} \models \alpha \wedge \xi_4$. Let $b$ be an
element of $\mf{B}$ that is not in $\mf{A}$. Then there are two cases,
in both which we show that $\mf{B} \models \psi_k$, as below:
\begin{enumerate}[nosep]
\item \vspace{2pt} $(\mf{B}, a_1, b, a_2) \models ((x \leq y) \wedge
  (y \leq z))$ for two elements $a_1, a_2$ of $\mf{A}$ such that
  $(\mf{A}, a_1, a_2) \models S(x, z)$; then $\mf{B} \models \neg
  \xi_3$ and hence $\mf{B} \models \psi_k$.
\item \vspace{2pt} $(\mf{B}, b) \models ((d \leq x) \vee (x \leq
  c))$. Since the interpretations of $c, d$ in $\mf{B}$ are resp.\ the
  same as those of $c, d$ in $\mf{A}$, we have $\mf{B} \models \neg
  \xi_2$ and hence $\mf{B} \models \psi_k$.
\end{enumerate}

We now show that $\varphi_k$ is not equivalent over $\cl{S}$ to any
$\exists^k \forall^*$ sentence. Towards a contradiction, suppose
$\varphi_k$ is equivalent over $\cl{S}$ to the sentence $\gamma :=
\exists x_1 \ldots \exists x_k$ $\forall^n \bar{y} \beta(x_1, \ldots,
x_k, \bar{y})$, where $\beta$ is a quantifier-free formula.  Consider
the structure $\mf{A} = (\mathsf{U}_{\mf{A}}, \leq^{\mf{A}},
S^{\mf{A}}, P^{\mf{A}}, c^{\mf{A}}, d^{\mf{A}})$, where the universe
$\mathsf{U}_{\mf{A}} = \{1, \ldots, (8n+1)\times (k+1)\}$,
$\leq^{\mf{A}}$ and $S^{\mf{A}}$ are respectively the usual linear
order and successor relation on $\mathsf{U}_{\mf{A}}$, $c^{\mf{A}} =
1, d^{\mf{A}} = (8n+1) \times (k+1)$ and $P^{\mf{A}} = \{(4n+1) + i
\times (8n +1) \mid i \in \{0, \ldots, k\} \}$. We see that $\mf{A}
\models (\xi_1 \wedge \xi_2 \wedge \xi_3 \wedge \xi_4 \wedge \neg
\xi_5)$ and hence $\mf{A} \models \varphi_k$. Then $\mf{A} \models
\gamma$. Let $a_1, \ldots, a_k$ be the witnesses in $\mf{A}$ to the
$k$ existential quantifiers of $\gamma$.

It is clear that there exists $i^* \in \{0, \ldots, k\}$ such that
$a_j$ does not belong to $\{(8n+1)\times i^* + 1, \ldots, (8n+1)\times
(i^*+1) \}$ for each $j \in \{1, \ldots, k\}$.  Then consider the
structure $\mf{B}$ that is identical to $\mf{A}$ except that
$P^{\mf{B}} = P^{\mf{A}} \setminus \{(4n+1) + i^* \times (8n+1)\}$.
It is clear from the definition of $\mf{B}$ that $\mf{B} \models
(\xi_1 \wedge \xi_2 \wedge \xi_3 \wedge \xi_4 \wedge \xi_5)$ and hence
$\mf{B} \models \neg \varphi_k$. We now show a contradiction by
showing that $\mf{B} \models \gamma$.

We show that $\mf{B} \models \gamma$ by showing that $(\mf{B}, a_1,
\ldots, a_k) \models \forall^n \bar{y} \beta(x_1, \ldots, x_k,
\bar{y})$.  This is in turn done by showing that for any $n$-tuple
$\bar{e} = (e_1, \ldots, e_n)$ from $\mf{B}$, there exists an
$n$-tuple $\bar{f} = (f_1, \ldots, f_n)$ from $\mf{A}$ such that the
(partial) map $\rho: \mf{B} \rightarrow \mf{A}$ given by $\rho(1) =
1$, $\rho((8n+1)\times (k+1)) = (8n+1)\times (k+1)$, $\rho(a_j) = a_j$
for $j \in \{1, \ldots, k\}$ and $\rho(e_j) = f_j$ for $j \in \{1,
\ldots, n\}$ is such that $\rho$ is a partial isomorphism from
$\mf{B}$ to $\mf{A}$. Then since $(\mf{A}, a_1, \ldots, a_k) \models
\forall^n \bar{y} \beta(x_1, \ldots, x_k, \bar{y})$, we have $(\mf{A},
a_1, \ldots, a_k, \bar{f}) \models \beta(x_1, \ldots, x_k, \bar{y})$
whereby $(\mf{B}, a_1, \ldots, a_k, \bar{e}) \models \beta(x_1,
\ldots, x_k, \bar{y})$. As $\bar{e}$ is an arbitrary $n$-tuple from
$\mf{B}$, we have $(\mf{B}, a_1, \ldots, a_k) \models \forall^n
\bar{y} \beta(x_1, \ldots, x_k, \bar{y})$.

Define a \emph{contiguous segment in $\mf{B}$} to be a set of $l$
distinct elements of $\mf{B}$, for some $l \ge 1$, that are contiguous
w.r.t. the linear ordering in $\mf{B}$. That is, if $b_1, \ldots, b_l$
are the distinct elements of the aforesaid contiguous segment, then
$(b_j, b_{j+1}) \in S^{\mf{B}}$ for $j \in \{1, \ldots, l-1\}$. We
represent such a contiguous segment as $\left[b_1, b_l\right]$, and
view it as an interval in $\mf{B}$.  Given an $n$-tuple $\bar{e}$ from
$\mf{B}$, a \emph{contiguous segment of $\bar{e}$ in $\mf{B}$} is a
contiguous segment in $\mf{B}$, all of whose elements belong to (the
set underlying) $\bar{e}$.  A \emph{maximal contiguous segment of
  $\bar{e}$ in $\mf{B}$} is a contiguous segment of $\bar{e}$ in
$\mf{B}$ that is not strictly contained in another contiguous segment
of $\bar{e}$ in $\mf{B}$. Let $\mathsf{CS}$ be the set of all maximal
contiguous segments of $\bar{e}$ in $\mf{B}$. Let $\mathsf{CS}_1
\subseteq \mathsf{CS}$ be the set of all those segments of
$\mathsf{CS}$ that have an intersection with the set $\{1, \ldots,
(8n+1)\times i^*\} \cup \{ (8n+1)\times (i^*+1) + 1, \ldots,
(8n+1)\times(k+1) \}$. Let $\mathsf{CS}_2 = \mathsf{CS} \setminus
\mathsf{CS}_1$.  Then all intervals in $\mathsf{CS}_2$ are contained
in the interval $\left[(8n+1) \times i^* + 1, (8n+1) \times (i^* + 1)
  \right]$.  Let $\mathsf{CS}_2 = \{\left[i_1, j_1\right], \left[i_2,
  j_2\right] \ldots, \left[i_r, j_r\right]\}$ such that $i_1 \leq j_1
< i_2 \leq j_2 < \ldots < i_r \leq j_r$.  Observe that $r \leq n$.
Let $\mathsf{CS}_3$ be the set of contiguous segments in $\mf{A}$
defined as $\mathsf{CS}_3 = \{\left[i_1', j_1'\right], \left[i_2',
  j_2'\right], \ldots, \left[i_r', j_r'\right]\}$ where $i_1' =
(8n+1)\times i^* + n + 1, j_1' = i_1' + (j_1 - i_1)$, and for $2 \leq
l \leq r$, we have $i_l' = j_{l-1}' + 2$ and $j_l' = i_l' + (j_l -
i_l)$. Observe that the sum of the lengths of the segments of
$\mathsf{CS}_2$ is at most $n$, so that $j_r' \leq (8n+1)\times i^* +
3n + 1$.

Now consider the tuple $\bar{f} = (f_1, \ldots, f_n)$ defined using
$\bar{e} = (e_1, \ldots, e_n)$ as follows. Let
$\text{Elements}(\mathsf{CS}_1)$,
resp. $\text{Elements}(\mathsf{CS}_2)$, denote the elements contained
in the segments of $\mathsf{CS}_1$, resp. $\mathsf{CS}_2$.  For $1
\leq l \leq n$, if $e_l \in \text{Elements}(\mathsf{CS}_1)$, then $f_l
= e_l$. Else suppose $e_l$ belongs to the segment $\left[i_s,
  j_s\right]$ of $\mathsf{CS}_2$ where $1 \leq s \leq r$, and suppose
that $e_l = i_s + t$ for some $t \in \{0, \ldots, (j_s - i_s)\}$. Then
choose $f_l = i_s' + t$.  We now verify that the (partial) map $\rho:
\mf{B}$ $ \rightarrow \mf{A}$ that is given by $\rho(1) = 1$,
$\rho((8n+1)\times (k+1)) = (8n+1)\times (k+1)$, $\rho(a_j) = a_j$ for
$j \in \{1, \ldots, k\}$ and $\rho(e_l) = f_l$ for $l \in \{1, \ldots,
n\}$, is indeed a partial isomorphism from $\mf{B}$ to $\mf{A}$.
\end{proof}

%% file: conclusion-2.0.tex
\section{Conclusion and future directions}\label{section:conclusion}

In this paper, we presented a new proof of the extensional form of the
generalized {\lt} theorem ($\glt{k}$) for theories first shown
in~\cite{abhisekh-apal}, and thereby obtained a new proof of the
theorem for sentences in both its forms substructural and
extensional. Our proof avoids using $\lambda$-saturation as used
in~\cite{abhisekh-apal}, and instead constructs structures with just
the ``needed saturation'' to prove the theorem. As our second result,
we presented a strengthening of the failure of the {\lt} theorem in
the finite by showing that not only universal sentences, but even
$\exists^k \forall^*$ sentences for any fixed $k$ are not expressive
enough to capture the semantic property of hereditariness in the
finite.

We now mention two future directions concerning our results. The first
is in connection with the {\lt} theorem in the finite. The
counterexample to this theorem in the finite as presented
in~\cite{gurevich-alechina} uses two binary relations and two
constants. But what happens if the vocabulary contains only one binary
relation and some constants/unary relations? There are positive
results shown when the binary relation is constrained to be
interpreted as special kinds of posets, specifically linear orders or
(more generally) poset-theoretic trees, or special kinds of graphs,
specifically subclasses of bounded clique-width graphs such as classes
of bounded tree-depth/shrub-depth and $m$-partite
cographs~\cite{abhisekh-thesis-short}. (In fact, over all these
classes, even $\glt{k}$ is true for all $k$.) But the case of an
unconstrained binary relation remains open, motivating the following
question.\footnote{We remark that special cases
  of~\ref{problem:los-tarski-finite} for different fragments of FO
  were studied in~\cite{psc(b)-arxiv-may-2012} and the {\lt} theorem
  was verified to be true for these cases over all finite structures,
  even for arbitrary finite relational vocabularies.}

\begin{problem}\label{problem:los-tarski-finite}
Is the (relativized version of the) {\lt} theorem true over all finite
colored directed graphs? The same question also for undirected graphs.
\end{problem}

Our second future direction concerns $\glt{k}$ over all finite
structures. Theorem~\ref{prop:strong-failure-of-LT-in-the-finite}
exhibits for each $k$, a sentence $\varphi_k$ that is hereditary over
all finite structures but that is not equivalent over this class to
any $\exists^k \forall^*$ sentence. We however observe that
$\varphi_k$ is itself equivalent to an $\exists^{k+1} \forall^*$
sentence. So that this counterexample to $\glt{k}$ is not a
counterexample to $\glt{k+1}$. This raises the natural question of
whether all counterexamples to $\glt{k}$ in the finite, are simply
$\Sigma^0_2$ sentences, or sentences equivalent to these. Given that
any $\Sigma^0_2$ sentence is $k$-hereditary for some $k$, we pose the
aforesaid question as the following problem.
\begin{problem}\label{problem:glt-k-finite}
Is it the case that over the class of all finite structures, a
sentence is $k$-hereditary for some $k$ if, and only if, it is
equivalent to a $\Sigma^0_2$ sentence?
\end{problem}

Observe that the version of Problem~\ref{problem:glt-k-finite} in
which arbitrary structures are considered instead of finite
structures, has a positive answer due to Theorem~\ref{theorem:glt-k}
(which is a stronger statement). Much like the {\lt} theorem, results
from classical model theory almost invariably fail in the
finite~\cite{gurevich-alechina}. Resolving
Problem~\ref{problem:glt-k-finite} in the affirmative would then give
us a preservation theorem that survives passage to all finite
structures.